\documentclass[9pt,shortpaper,twoside,web]{ieeecolor}

\usepackage{generic}
\usepackage{cite}
\usepackage{algorithmic}
\usepackage{textcomp}
\usepackage{graphicx}
\usepackage{times}
\usepackage{amsmath}
\usepackage{amsfonts}
\usepackage{amssymb}
\usepackage{float}
\usepackage{balance}
\usepackage{bm}
\usepackage{color}
\usepackage[final]{microtype}
\usepackage{mathtools}
\mathtoolsset{showonlyrefs}

\newtheorem{theorem}{\bf Theorem}
\newtheorem{proposition}{\bf Proposition}

\newtheorem{definition}{\bf Definition}
\newtheorem{remark}{\bf Remark}
\newtheorem{example}{\bf Example}
\newtheorem{lemma}{\bf Lemma}
\newtheorem{corollary}{\bf Corollary}

\definecolor{darkgray}{gray}{0.2}
\definecolor{hdarkgray}{gray}{0.4}
\definecolor{llightgray}{gray}{0.6}
\definecolor{lightgray}{gray}{0.8}

\newcommand{\abs}[1]{\left\lvert#1\right\rvert}

\newcommand{\norm}[1]{\left\lvert\left\lvert#1\right\rvert\right\rvert}
\newcommand{\barpow}[1]{\left\lfloor#1\right\rceil}
\newcommand{\barpowvec}[1]{\abs{\barpow{#1}}}

\def\BibTeX{{\rm B\kern-.05em{\sc i\kern-.025em b}\kern-.08em
    T\kern-.1667em\lower.7ex\hbox{E}\kern-.125emX}}
\begin{document}
\title{A Lyapunov-like Characterization of Predefined-Time Stability}
\author{Esteban~Jim\'enez-Rodr\'iguez, 
    	Aldo~Jonathan~Mu\~noz-V\'azquez, 
        Juan~Diego~S\'anchez-Torres, 
        Michael~Defoort and~Alexander~G.~Loukianov
\thanks{Submitted on: 11/19/2018.}
\thanks{E.~Jim\'enez-Rodr\'iguez and A.~G.~Loukianov are with the Department of Electrical Engineering, Cinvestav-Guadalajara, Zapopan, Jalisco, 45019 M\'exico (e-mail: \{ejimenezr, louk\}@gdl.cinvestav.mx). }
\thanks{A.~J.~Mu\~noz-V\'azquez is with CONACYT--M\'exico, and with the School of Engineering at the Autonomous University of Chihuahua, Chihuahua, Chihuahua, 31100 M\'exico (e-mail: aldo.munoz.vazquezz@gmail.com).}
\thanks{J.~D.~S\'anchez-Torres is with the Research Laboratory on Optimal Design, Devices and Advanced Materials -OPTIMA-, Department of Mathematics and Physics, ITESO, Tlaquepaque, Jalisco, 45604 M\'exico (e-mail: dsanchez@iteso.mx).}
\thanks{M.~Defoort is with the LAMIH, UMR CNRS 8201, Polytechnic University of Hauts-de-France, Valenciennes, 59313 France (e-mail: michael.defoort@uphf.fr).}}

\maketitle

\begin{abstract} This technical note studies Lyapunov-like conditions to ensure a class of dynamical systems to exhibit predefined-time stability. The origin of a dynamical system is predefined-time stable if it is fixed-time stable and an upper bound of the settling-time function can be \textit{arbitrarily} chosen a priori through a suitable selection of the system parameters.
We show that the studied Lyapunov-like conditions allow to demonstrate equivalence between previous Lyapunov theorems for predefined-time stability for autonomous systems. Moreover, the obtained Lyapunov-like theorem is extended for analyzing the property of predefined-time ultimate boundedness with predefined bound, which is useful when analyzing uncertain dynamical systems.
Therefore, the proposed results constitute a general framework for analyzing predefined-time stability, and they also unify a broad class of systems which present the predefined-time stability property. On the other hand, the proposed framework is used to design robust controllers for affine control systems, which induce predefined-time stability (predefined-time ultimate boundedness of the solutions) w.r.t. to some desired manifold. 
A simulation example is presented to show the behavior of a developed controller, especially regarding the settling time estimation.
\end{abstract}


\section{Introduction}
\label{sec:introduction}
The development of control, observation, and optimization algorithms for many industrial applications requires the fulfillment of certain time-response constraints in order to comply with a certain quality or safety standards. To deal with these requirements, several developments concerning the \textit{finite-time stability} concept have been carried out in~\cite{Roxin1966,Haimo1986,Utkin1992,Bhat2000,Moulay2008}.
Nevertheless, the finite settling time provided by finite-time convergent algorithms is usually an unbounded function of the system's initial conditions. This concern gives rise to a stronger form of stability called \textit{fixed-time stability}, where the settling-time function is bounded. The notion of fixed-time stability has been investigated in~\cite{Andrieu2008,Cruz-Zavala2010,Polyakov2012}.

Although the concept of fixed-time stability represents a significant advantage over the concept of finite-time stability because of the boundedness of the settling time, it cannot be guaranteed in general that the convergence time can be arbitrarily selected through the system tunable parameters. To overcome this mentioned drawback, it is necessary to consider another class of dynamical systems that exhibit the property of \textit{predefined-time stability}, which has been studied in~\cite{Sanchez-Torres2018,Jimenez-Rodriguez2017}. For these systems, an upper bound of the settling-time function can be arbitrarily chosen through an appropriate selection of the system parameters.

On the other hand, Lyapunov methods have proved to be a handy tool for analyzing and designing nonlinear control systems~\cite{Bacciotti2005,Khalil2000}. In the same manner, they have been highly used for convergence rate estimation in systems exhibiting finite- and fixed-time stability properties~\cite{Bhat2000,Polyakov2012,Defoort2016}, and in particular in systems with sliding modes~\cite{Polyakov2014}.
For systems exhibiting the predefined-time stability property, similar methods have been applied. For instance, different Lyapunov-like theorems for predefined-time stability were proposed in~\cite{Jimenez-Rodriguez2018a,Sanchez-Torres2018,Sanchez-Torres2018a,Aldana-Lopez2019} allowing the development of several control applications~\cite{Jimenez-Rodriguez2018,Munoz-Vazquez2019,Sanchez-Torres2019}. 

This paper investigates Lyapunov-like sufficient conditions for predefined-time stability of autonomous systems. The derived Lyapunov-like theorem allows to establish equivalence with several previous Lyapunov-like theorems~\cite{Sanchez-Torres2018,Sanchez-Torres2018a,Aldana-Lopez2019}, unifying all the past contributions in predefined-time stability for autonomous systems. Moreover, this framework is extended to the analysis of predefined-time ultimate boundedness, which is specially usefull when analyzing uncertain systems.
To demonstrate the applicability of the proposed framework, the developed results are used to design a family of continuous (respectively, discontinuous) controllers, which ensure predefined-time ultimate boundedness of the solutions to an arbitrarily small vicinity of a desired manifold (respectively, ensure predefined-time stability to a desired manifold). Finally, all the mentioned properties are validated through a simulation example, in order to show the behavior of the proposed controller, especially regarding the settling time estimation.

\section{Preliminaries}\label{sec:prelim}

\subsection{Notation} \label{subsec:notation}
Throughout the paper, the following notation is prevalent:
\begin{itemize}
\item $\mathbb{R}$ is the set of real numbers, $\mathbb{R}_+=\{x\in\mathbb{R}\,:\,x>0\}$, $\mathbb{R}_{\geq 0}=\{x\in\mathbb{R}\,:\,x\geq 0\}$ and $\Bar{\mathbb{R}}_{+}=\mathbb{R}_{+}\cup\{\infty\}$.
\item For $\bm{x}\in\mathbb{R}^n$, $\bm{x}^T$ denotes its transpose, $\norm{\bm{x}}=\sqrt{\bm{x}^T\bm{x}}$ and, for $r\in\mathbb{R}_+$, $B_r(\bm{x})=\{\bm{y}\in\mathbb{R}^n\,:\,\norm{\bm{y}-\bm{x}}<r\}$.
\item For any real number $h$, the functions $\barpow{\cdot}^h:\mathbb{R}\to\mathbb{R}$ and $\barpowvec{\bm{\cdot}}^h:\mathbb{R}^n\to\mathbb{R}^n$ are defined as $\barpow{x}^h=|x|^h\mathrm{sign}(x)$ for any $x\in \mathbb{R}\setminus\{0\}$ and $\barpowvec{\bm{x}}^{h}=\frac{\bm{x}}{\norm{\bm{x}}^{1-h}}$ for any $\bm{x}\in \mathbb{R}^n\setminus\{\bm{0}\}$, respectively. Moreover, if $h>0$, $\barpow{0}^h=0$ and $\barpowvec{\bm{0}}^{h}=\bm{0}$.
\item Whereas $\dot{\bm{x}}=\frac{d \bm{x}}{d t}$ denotes the first derivative of the function $\bm{x}:\mathbb{R}\to\mathbb{R}^n$ with respect to time, $\theta'(z) = \frac{d\theta}{dz}$ denotes the first derivative of the function $\theta:\mathbb{R}\to\mathbb{R}$ with respect to the variable $z$, which may represent anything but the time variable $t$.
\item For $\alpha,\beta\in\mathbb{R}_+$, $\Gamma(\alpha)=\int_0^\infty t^{\alpha-1}e^{-t}\text{d}t$ is the Gamma Function and $\mathcal{B}(\alpha,\beta)=\int_0^1 t^{\alpha-1}(1-t)^{\beta-1}\text{d}t$ is the Beta Function; additionally, $\gamma(\alpha,r)=\int_0^r t^{\alpha-1}e^{-t}\text{d}t$ and $P(\alpha,r)=\frac{\gamma(\alpha,r)}{\Gamma(\alpha)}$ are the Incomplete Gamma Function and the regularized Incomplete Gamma Function, respectively, which are defined for all $r\in\mathbb{R}_{\geq 0}$; finally $b(\alpha,\beta,r)=\int_0^{r} t^{\alpha-1}(1-t)^{\beta-1}\text{d}t$ and $I(\alpha,\beta,r)=\frac{b(\alpha,\beta,r)}{\mathcal{B}(\alpha,\beta)}$ are the Incomplete Beta Function and the regularized Incomplete Beta Function, respectively, which are defined for all $r\in[0,1]$~\cite{Abramowitz1964}.
\end{itemize}

\subsection{On predefined-time stability} \label{subsec:fix_pred}

Consider the following autonomous system:
\begin{equation}\label{eq:sys}
\dot{\bm{x}}=\bm{f}(\bm{x};\bm{\rho}), \ \ \bm{x}(0)=\bm{x}_0,
\end{equation}
where $\bm{x}:\mathbb{R}_{\geq 0}\to\mathbb{R}^n$ is the system state, the vector $\bm{\rho}\in\mathbb{R}^l$ stands for the \textit{tunable} parameters of~\eqref{eq:sys}.
The function $\bm{f}:\mathbb{R}^n\rightarrow\mathbb{R}^n$ may be discontinuous, and such that the solutions of~\eqref{eq:sys} exist and are unique in the sense of Filippov (see~\cite{Filippov1988} and~\cite[Proposition 5]{Cortes2008}). Thus, $\bm{\Phi}(t,\bm{x}_0)$ denotes the solution of~\eqref{eq:sys} starting from $\bm{x}_0\in\mathbb{R}^n$ at $t=0$. Moreover, the origin $\bm{x}=\bm{0}$ is the unique equilibrium point of~\eqref{eq:sys}.



\begin{remark} The parameter dependent system~\eqref{eq:sys} is equivalent to the controlled system 
\begin{equation}\label{eq1}
\dot{\bm{x}}=\bm{g}(\bm{x},\bm{u}), 
\end{equation}
where $\bm{g}:\mathbb{R}^n\times\mathbb{R}^m \to \mathbb{R}^n$, the control $\bm{u}\in\mathbb{R}^m$ is a feedback function of $\bm{x}$ with tunable parameters $\bm{\rho}$, i.e., $\bm{u}=\bm{\phi}(\bm{x};\bm{\rho})$, with $\bm{\phi}:\mathbb{R}^n\to\mathbb{R}^m$. Substituting  $\bm{u}=\bm{\phi}(\bm{x};\bm{\rho})$ in \eqref{eq1} eliminates $\bm{u}$ and yields $\bm{f}(\bm{x};\bm{\rho}):=\bm{g}(\bm{x},\bm{\phi}(\bm{x};\bm{\rho}))$.
\end{remark}

All the notions defined and treated hereafter are global, so we will omit to indicate it.

\begin{definition}[Stability notions~\cite{Lopez-Ramirez2018}] \label{def:stability} The origin of~\eqref{eq:sys} is said to be
\begin{itemize}
\item \textbf{Lyapunov stable} if for any $\bm{x}_0\in\mathbb{R}^n$, the solution $\bm{\Phi}(t,\bm{x}_0)$ is defined for all $t\geq 0$, and for any $\epsilon>0$, there is $\delta>0$ such that for any $\bm{x}_0\in\mathbb{R}^n$, if $\bm{x}_0\in B_{\delta}(\bm{0})$ then $\bm{\Phi}(t,\bm{x}_0)\in B_{\epsilon}(\bm{0})$ for all $t\geq0$;
\item \textbf{asymptotically stable} if it is Lyapunov stable and $\bm{\Phi}(t,\bm{x}_0)\to 0$ as $t\to \infty$, for any $\bm{x}_0\in\mathbb{R}^n$;
\item \textbf{finite-time stable} if it is Lyapunov stable and for any $x_0\in\mathbb{R}^n$ there exists $0\leq \tau <\infty$ such that $\bm{\Phi}(t,\bm{x}_0)=\bm{0}$ for all $t\geq \tau$. The function $T(\bm{x}_0)=\inf \left\lbrace \tau \geq 0:\bm{\Phi}(t,\bm{x}_0)=\bm{0},\, \forall t\geq \tau \right\rbrace$ is called the \textbf{settling-time function} of~\eqref{eq:sys};
\item \textbf{fixed-time stable} if it is finite-time stable and the settling-time function of~\eqref{eq:sys}, $T(\bm{x}_0)$, is bounded on $\mathbb{R}^n$, i.e. there exists $T_{\max}$ such that $\sup_{\bm{x}_0\in\mathbb{R}^n}T(\bm{x}_0)\leq T_{\max}<\infty$
\end{itemize}
\end{definition}

\begin{example}\label{exmp:fixnopred} Consider system
\begin{align}\label{eq:fixnopred}
\begin{split}
\dot{x} 
        &= -\frac{1}{\rho_1}\barpow{x}^{\rho_2}-\rho_1\barpow{x}^{2-\rho_2},
\end{split}
\end{align}
where $x\in\mathbb{R}$ is the system state, $\bm{\rho}=\left[\rho_1,\,\rho_2\right]^T\in\mathbb{R}^2$ is the vector of tunable parameters of~\eqref{eq:fixnopred}, which comply to $\rho_1>0$ and $0<\rho_2<1$. Using \cite[Lemma~1]{Polyakov2012}, one can easily show that the origin of~\eqref{eq:fixnopred} is fixed-time stable. Moreover, from~\cite[Theorem~1]{Aldana-Lopez2019}, the settling-time function of~\eqref{eq:fixnopred} satisfies
\[\sup_{x_0\in\mathbb{R}} T(x_0) = \frac{\mathcal{B}\left(1/2,1/2\right)}{2\rho_1^{-1/2}\rho_1^{1/2}(1-\rho_2)}=\frac{\pi}{2(1-\rho_2)}>\frac{\pi}{2}.\]
\end{example}

This example shows that the convergence time for system~\eqref{eq:fixnopred}, whose origin is fixed-time stable, cannot be reduced arbitrarily no matter how the parameters $\bm{\rho}$ are tuned. The case when the convergence time can be arbitrarily assigned through an appropriate tuning of the   system parameters $\bm{\rho}$ corresponds to the notion of predefined-time stability, which is defined as follows:

\begin{definition}\label{def:predefined} The origin of~\eqref{eq:sys} is said to be \textbf{predefined-time stable} if it is fixed-time stable and for any $T_c\in\mathbb{R}_+$, there exists some $\bm{\rho}\in\mathbb{R}^l$ such that the settling-time function of~\eqref{eq:sys} satisfies \[\sup_{\bm{x}_0\in\mathbb{R}^n}T(\bm{x}_0)\leq T_c.\]
\end{definition}

\begin{example}\label{exmp:pred} Consider system (see~\cite{Polyakov2012,Andrieu2008})
\begin{align}\label{eq:pred}
\begin{split}
\dot{x} 
        &= -\barpow{\rho_1\barpow{x}^{\rho_3}+\rho_2\barpow{x}^{\rho_4}}^{\rho_5},
\end{split}
\end{align}
where $x\in\mathbb{R}$ is the state of the system, $\bm{\rho}=\left[\rho_1,\,\rho_2,\,\rho_3,\,\rho_4,\,\rho_5\right]^T\in\mathbb{R}^5$ is the vector of tunable parameters of~\eqref{eq:pred}, which comply to $\rho_1,\rho_2,\rho_5>0$ and $0<\rho_5\rho_3<1<\rho_5\rho_4$. The origin of~\eqref{eq:pred} is fixed-time stable, by~\cite[Lemma~1]{Polyakov2012}. Moreover, given $T_c\in\mathbb{R}_+$, there exist $\rho_1=\rho_2=\frac{\Gamma(1/4)^4}{4\pi T_c^2}$, $\rho_3=1$, $\rho_4=3$ and $\rho_5=\frac{1}{2}$, such that the settling-time function of~\eqref{eq:pred} satisfies (see~\cite[Theorem~1]{Aldana-Lopez2019}) \[\sup_{x_0\in\mathbb{R}} T(x_0)=\frac{\Gamma(1/4)^2}{\left(\frac{\Gamma(1/4)^4}{4\pi T_c^2}\right)^{1/2}\Gamma(1/2)(3-1)} = T_c.\]
Thus, the origin of system~\eqref{eq:pred} is, in fact, predefined-time stable.
\end{example}

The following proposition is an immediate consequence of Definition~\ref{def:predefined}, of predefined-time stability.

\begin{proposition}\label{prop:fix_no_pred} If a system does not have tunable parameters, then its origin is not predefined-time stable.
\end{proposition}

From Proposition~\ref{prop:fix_no_pred}, every system with fixed-time stable origin whose parameters are fixed numerical values (they are not tunable), cannot exhibit the predefined-time stability property.

On the other hand, when dealing with systems subject to uncertainties or external perturbations, it is difficult or even impossible to ensure exact convergence to the origin. Instead, it is common to provide convergence not to the origin but to a vicinity of it. In this sense, it would be useful to ensure that, not only the convergence time can be arbitrarily assigned, but also that the radius of the vicinity can be arbitrarily selected through an appropriate tuning of the parameters of the system. This notion is formally defined as follows:

\begin{definition}\label{def:ptub} A solution  $\bm{\Phi}(t,\bm{x}_0)$ of~\eqref{eq:sys} is said to be \textbf{predefined-time ultimately bounded with predefined bound} if for any $T_c,b\in\mathbb{R}_+$, there exist some $\bm{\rho}\in\mathbb{R}^l$ such that for any $\bm{x}_0\in\mathbb{R}^n$, $\norm{\bm{\Phi}(t,\bm{x}_0)}\leq b$ for all $t\geq T_c$.
\end{definition}

\subsection{Class $\mathcal{K}^1$ functions}
Inspired in the class $\mathcal{K}$ functions in~\cite[Definition 1]{Kellett2014} and~\cite[Definition 4.2]{Khalil2000}, the class $\mathcal{K}^1$ functions are defined as follows:

\begin{definition} [$\mathcal{K}^1$ functions] \label{def:class_K1} A scalar continuous function $\kappa:\mathbb{R}_{\geq 0}\to\left[0,1\right)$ is said to belong to class $\mathcal{K}^1$, denoted as $\kappa\in\mathcal{K}^1$, if it is strictly increasing, $\kappa(0)=0$ and $\kappa(r)\to 1$ as $r\to\infty$.
\end{definition}

If $\kappa\in\mathcal{K}^1$ is also differentiable, it is said to be a differentiable class $\mathcal{K}^1$ function. In such a case, there exists a continuous function $\Phi:\mathbb{R}_{\geq 0}\to\mathbb{R}_+$ such that $\frac{d\kappa}{dr}=\Phi(r)>0$.

The above can be equivalently written in an integral form as $\kappa(r)=\int_{0}^{r}\Phi(z)\text{d}z$. Since $\kappa(r)\to 1$ as $r\to\infty$, the function $\Phi$ is required to satisfy $\int_0^\infty \Phi(z)\text{d}z=1$, i.e., functions $\Phi$ and $\kappa$ can be viewed as probability density functions and cumulative distribution functions, respectively, of positive random variables.

\begin{proposition} \label{prop:K1_biject} Every class $\mathcal{K}^1$ function is bijective.
\end{proposition}
\begin{proof} Let $\kappa\in\mathcal{K}^1$. It is injective because $\kappa$ is continuous and strictly increasing. Moreover, since $\kappa(0)=0$, $\kappa$ is continuous and strictly increasing, and $\lim_{r\to\infty}\kappa(r)=1$, its image is $\kappa(\mathbb{R}_{\geq 0})=[0,1)$. Thus, it is surjective. Hence, it is concluded that $\kappa$ is bijective.
\end{proof}

Since every class $\mathcal{K}^1$ function is bijective, their inverse exist.


\begin{proposition} Let $\kappa\in\mathcal{K}^1$. Then:
\begin{itemize}
\item[\textit{(i)}] $\kappa^{-1}(0)=0$;
\item[\textit{(ii)}] $\kappa^{-1}$ is continuous (i.e., $\kappa$ is a homeomorphism) and strictly increasing.
\item[\textit{(iii)}] $\lim_{r\to 1^{-}}\kappa^{-1}(r)=\infty$.
\end{itemize}
\end{proposition}
\begin{proof} These properties follow directly from Definition~\ref{def:class_K1}, and the fact that the inverse of a strictly increasing function is continuous and strictly increasing.
\end{proof}

The next lemma states some useful properties of class $\mathcal{K}_\infty$ and class $\mathcal{K}^1$ functions, which will be used in the next section.
\begin{lemma}\label{lem:comparison} Let $\alpha\in\mathcal{K}_\infty$ (see~\cite[Definition 1]{Kellett2014}) and $\kappa_1,\kappa_2\in\mathcal{K}^1$. Then, $\kappa_1\circ\alpha\in\mathcal{K}^1$, and $\kappa_1^{-1}\circ\kappa_2\in\mathcal{K}_\infty$.
\end{lemma}
\begin{proof} The composition of increasing functions is increasing. Moreover, note that $(\kappa_1\circ\alpha)(0)=\kappa_1(\alpha(0))=\kappa_1(0)=0$ and  $(\kappa_1^{-1}\circ\kappa_2)(0)=\kappa_1^{-1}(\kappa_2(0))=\kappa_1^{-1}(0)=0$. Finally, since $\kappa_1$ is an homeomorphism, \[\lim_{r\to\infty}(\kappa_1\circ\alpha)(r)=\kappa_1\left(\lim_{r\to\infty}\alpha(r)\right)=1,\] and \[\lim_{r\to\infty}(\kappa_1^{-1}\circ\kappa_2)(r)=\kappa_1^{-1}\left(\lim_{r\to\infty}\kappa_2(r)\right)=\infty.\]
\end{proof}

\begin{example}\label{exm:k_examples} Some examples of $\mathcal{K}^1$ functions are:
\begin{itemize}
\item [\textit{(i)}] $\kappa(r)=1-\exp(-r)$;
\item [\textit{(ii)}] $\kappa(r)=\frac{2}{\pi}\arctan(r)$;
\item [\textit{(iii)}] $\kappa(r)=\frac{r}{r+\alpha}$, with $\alpha>0$;
\item [\textit{(iv)}] $\kappa(r)=P(\alpha,r)$, with $\alpha>0$ (see Subsection~\ref{subsec:notation}).
\item [\textit{(v)}] $\kappa(r)=I(\alpha,\beta,\frac{r}{r+1})$, with $\alpha,\beta>0$ (see Subsection~\ref{subsec:notation}).
\end{itemize}
\end{example}

\section{A Lyapunov characterization of predefined-time stability}\label{sec:main}

This section states a  Lyapunov-like theorem for predefined-time stability.  The importance of this theorem is twofold. On the one hand, it establishes equivalence between previous Lyapunov theorems for predefined-time stability, constituting a unifying result. On the other hand, it allows designing predefined-time stable controllers, as shown in Section~\ref{sec:control}. Moreover, this Lyapunov theorem is extended for analyzing the property of predefined-time ultimate boundedness with predefined bound. Consequently, the results presented in this section constitute the main contribution of this note.

\begin{theorem}\label{thm:predefined} Let $\kappa\in\mathcal{K}^1$ be differentiable in $\mathbb{R}\setminus \{0\}$, and $V:\mathbb{R}^n\to\mathbb{R}_{\geq 0}$ be a continuous, positive definite and radially unbounded function. 
If for any $T_c\in\mathbb{R}_+$, there exists some $\bm{\rho}\in\mathbb{R}^l$, such that the time-derivative of $V$ along the trajectories of~\eqref{eq:sys} satisfies
\begin{equation}\label{eq:Vineq}
\dot{V}(\bm{x})\leq -\frac{1}{(1-p)T_c}\frac{\kappa(V(\bm{x}))^p}{\kappa'(V(\bm{x}))}, \quad \text{for } \bm{x}\in\mathbb{R}^n\setminus\left\lbrace\bm{0}\right\rbrace,
\end{equation}
for some $0\leq p<1$, then the origin of~\eqref{eq:sys} is predefined-time stable. Moreover, if~\eqref{eq:Vineq} is an equality, then $\sup_{\bm{x}_0\in\mathbb{R}^n}T(\bm{x}_0)=T_c$.
\end{theorem}

\begin{proof} Let $T_c\in\mathbb{R}_+$. Then, there exists some $\bm{\rho}\in\mathbb{R}^l$ such that~\eqref{eq:Vineq} holds. Moreover, since $V:\mathbb{R}^n\to\mathbb{R}_{\geq 0}$ is a continuous, positive definite and radially unbounded function, and its time-derivative~\eqref{eq:Vineq} is negative for $\bm{x}\in\mathbb{R}^n\setminus\left\lbrace\bm{0}\right\rbrace$, the origin of system~\eqref{eq:sys} is asymptotically stable~\cite{Khalil2000}.


Now, let $\bm{\Phi}(t,\bm{x}_0)$ be a solution of~\eqref{eq:sys} and let $y:\mathbb{R}_{\geq 0}\to\mathbb{R}_{\geq 0}$ be a function that satisfies 
\[
\dot{y} = -\frac{1}{(1-p)T_c}\frac{\kappa(y)^p}{\kappa'(y)},
\] 
and $V(\bm{x}_0)\leq y(0)$. Hence,
\[
\kappa(y(t))=\left\lbrace\begin{array}{cl}
    \left[\kappa(y(0))^{1-p}-\frac{t}{T_c}\right]^{\frac{1}{1-p}} & \text{if } 0\leq t\leq T_c\kappa(y(0))^{1-p} \\
    0 & \text{if } t>T_c\kappa(y(0))^{1-p},
\end{array}\right.
\]
and $V(\bm{\Phi}(t,\bm{x}_0))\leq y(t)$ (it is an equality only if~\eqref{eq:Vineq}, is an equality) by the comparison lemma~\cite{Khalil2000}. Thus, $V(\bm{\Phi}(t,\bm{x}_0))=0$ for $t \geq T_c\kappa(V(\bm{x}_0))^{1-p}$, implying that the trajectories of~\eqref{eq:sys} reach the origin in finite time, and the settling-time function satisfies 
\[
\sup_{\bm{x}_0\in\mathbb{R}^n}T(\bm{x}_0) \leq \sup_{\bm{x}_0\in\mathbb{R}^n} T_c\kappa(V(\bm{x}_0))^{1-p} = T_c.
\]
Hence, by Definition~\ref{def:predefined}, the origin of system~\eqref{eq:sys} is in fact predefined-time stable. Moreover, if~\eqref{eq:Vineq} is an equality, then $\sup_{\bm{x}_0\in\mathbb{R}^n}T(\bm{x}_0)=\sup_{\bm{x}_0\in\mathbb{R}^n}T_c \kappa(V(\bm{x}_0))^{1-p}=T_c$.
\end{proof}

\begin{remark} Theorem~\ref{thm:predefined} can be equivalently formulated in terms of a function $W:\mathbb{R}^n\rightarrow\mathbb{R}_{\geq 0}$, which satisfies $W(\bm{x})=0$ if and only if $\bm{x}=\bm{0}$, $0 \leq W(\bm{x}) < 1$, and $\lim_{\norm{\bm{x}}\to\infty}W(\bm{x})=1$. This equivalence is given by the transformation $W(\bm{x})=\kappa(V(\bm{x}))$. Moreover, in this equivalent reformulation, inequality~\eqref{eq:Vineq} is replaced by inequality
\[
\dot{W}(\bm{x}) \leq -\frac{1}{(1-p)T_c}W(\bm{x})^p, \quad \text{for } \bm{x}\in\mathbb{R}^n\setminus\left\lbrace\bm{0}\right\rbrace.
\]
In this sense, \cite[Theorem 10]{Jimenez-Rodriguez2018a} is a corollary of Theorem~\ref{thm:predefined}, which is obtained fixing $p=0$.
\end{remark}

\begin{remark} Although the result in Theorem~\ref{thm:predefined} is independent of the choice of $\kappa\in\mathcal{K}^1$, the form of the differential inequality~\eqref{eq:Vineq} strongly depends on the particular selection of this function.
Indeed, previous Lyapunov-like theorems for predefined-time stability reported in the literature are, in fact, particular forms of Theorem~\ref{thm:predefined}. For instance:
\begin{itemize}
\item[\textit{(i)}] \cite[Theorem~III.1]{Sanchez-Torres2018a} is obtained from Theorem~\ref{thm:predefined} with the particular selections of $\kappa(r)=P\left(\frac{1-\beta q}{s},\alpha r^s\right)=\frac{\gamma\left(\frac{1-\beta q}{s},\alpha r^s\right)}{\Gamma\left(\frac{1-\beta q}{s}\right)}$, with $\alpha,\beta,s,q>0$ and $\beta q<1$, and $p=0$. Thus, 
inequality~\eqref{eq:Vineq} then becomes 
\begin{equation}\label{eq:V_gamma}
\dot{V}(\bm x)\leq -\frac{\alpha^{\frac{\beta q-1}{s}}\Gamma\left(\frac{1-\beta q}{s}\right)}{sT_c}\exp\left(\alpha V(\bm x)^s\right)V(\bm x)^{\beta q}    
\end{equation}
for $\bm{x}\in\mathbb{R}^n\setminus\left\lbrace\bm{0}\right\rbrace$.
\item[\textit{(ii)}] At the same time, with $\alpha=\beta=1$, $q=1-s$ and $0<s\leq1$, inequality \eqref{eq:V_gamma} reduces to
\[\dot{V}(\bm x)\leq -\frac{1}{sT_c}\exp\left(V(\bm x)^s\right)V(\bm x)^{1-s}, \quad \text{for } \bm{x}\in\mathbb{R}^n\setminus\left\lbrace\bm{0}\right\rbrace,\]
which is equivalent to Theorem~\ref{thm:predefined} with the particular selection of $\kappa(r)=1-\exp(-r^s)$. In this form, the result presented in
\cite[Theorem~2.1]{Sanchez-Torres2018} is recovered.
\item[\textit{(iii)}] \cite[Theorem~3]{Aldana-Lopez2019} is retrieved from Theorem~\ref{thm:predefined} with the particular selections of $\kappa(r)=I\left(m_s,m_q,\frac{\beta r^{q-s}}{\beta r^{q-s}+\alpha}\right),$ with $m_s=\frac{1-ks}{q-s}>0$, $m_q=\frac{kq-1}{q-s}>0$, $\alpha,\beta,k>0$ and $0<ks<1<kq$, and $p=0$. Replacing the above picks into
inequality~\eqref{eq:Vineq}, it yields \[\dot{V}(\bm x)\leq -\frac{\zeta}{T_c}\left(\alpha V(\bm x)^{s}+\beta V(\bm x)^{q}\right)^{k}, \quad \text{for } \bm{x}\in\mathbb{R}^n\setminus\left\lbrace\bm{0}\right\rbrace,\] where $\zeta=\frac{\Gamma(m_s)\Gamma(m_q)}{\alpha^k\Gamma(k)(q-s)}\left(\frac{\alpha}{\beta}\right)^{m_s}$.
\end{itemize}
Thus, in this work, it is shown for the first time that all previous Lyapunov-like theorems for predefined-time stability of autonomous systems are equivalent.
\end{remark}

Lyapunov analysis can also be extended to show predefined-time ultimate boundedness with predefined bound of the solutions of~\eqref{eq:sys} (see Definition~\ref{def:ptub}), even if there is no equilibrium point at the origin. Sufficient conditions are stated in the following theorem:

\begin{theorem}\label{thm:predefined_ub} Let $\kappa\in\mathcal{K}^1$ be differentiable in $\mathbb{R}\setminus \{0\}$, and $V:\mathbb{R}^n\to\mathbb{R}_{\geq 0}$ be a continuous, positive definite and radially unbounded function.
If for any $T_c,\mu\in\mathbb{R}_+$, there exists some $\bm{\rho}\in\mathbb{R}^l$, such that the time-derivative of $V$ along the trajectories of~\eqref{eq:sys} satisfies
\begin{equation}\label{eq:Vineq_ub}
\dot{V}(\bm{x})\leq -\frac{1}{(1-p)T_c}\frac{\kappa(V(\bm{x}))^p}{\kappa'(V(\bm{x}))}, \quad \text{for } \norm{\bm{x}}\geq\mu,
\end{equation} 
then, for any $\bm{x}_0\in\mathbb{R}^n$ the solution $\bm{\Phi}(t,\bm{x}_0)$ of~\eqref{eq:sys} satisfies \[\norm{\bm{\Phi}(t,\bm{x}_0)}\leq b=\alpha_1^{-1}(\alpha_2(\mu)), \quad \text{for all }t\geq T_c,\]
where $\alpha_1,\alpha_2\in\mathcal{K}_{\infty}$. 

Moreover, if $V(\bm{x})=\alpha(\norm{\bm{x}})$, with $\alpha\in\mathcal{K}_\infty$, then $b=\mu$ in the above inequality. This is, the solutions of~\eqref{eq:sys} are predefined-time ultimately bounded with predefined bound.
\end{theorem}

\begin{proof} Let $T_c,\mu\in\mathbb{R}_+$. Then, there exists $\bm{\rho}\in\mathbb{R}^l$ such that~\eqref{eq:Vineq_ub} holds. Since $V$ is a continuous, positive definite and radially unbounded function, there exist $\alpha_1,\alpha_2\in\mathcal{K}_\infty$ such that $\alpha_1(\norm{\bm{x}}) \leq V(\bm{x}) \leq \alpha_2(\norm{\bm{x}})$~\cite[Lemma 4.3]{Khalil2000}.

Note that $\norm{\bm{x}}<\mu \Longleftrightarrow \alpha_2(\norm{\bm{x}})<\alpha_2(\mu) \Rightarrow V(\bm{x})<\alpha_2(\mu)$, i.e. the set $\left\{\bm{x}\in\mathbb{R}^n:\norm{\bm{x}}<\mu\right\}\subseteq\left\{\bm{x}\in\mathbb{R}^n:V(\bm{x})< \alpha_2(\mu)\right\}$, or equivalently $\left\{\bm{x}\in\mathbb{R}^n:V(\bm{x})\geq \alpha_2(\mu)\right\}\subseteq\left\{\bm{x}\in\mathbb{R}^n:\norm{\bm{x}}\geq\mu\right\}$. Hence, inequality~\eqref{eq:Vineq_ub} holds for $V(\bm{x})\geq \alpha_2(\mu)$.

The above implies that the set $\left\{\bm{x}\in\mathbb{R}^n:V(\bm{x})\leq \alpha_2(\mu)\right\}$ is positively invariant, since the derivative $\dot{V}(\bm{x})$ is negative in its boundary $\left\{\bm{x}\in\mathbb{R}^n:V(\bm{x})= \alpha_2(\mu)\right\}$.

Now, we show that all trajectories starting in the set $\left\{\bm{x}\in\mathbb{R}^n:V(\bm{x})\geq \alpha_2(\mu)\right\}$, must enter the set $\left\{\bm{x}\in\mathbb{R}^n:V(\bm{x})\leq \alpha_2(\mu)\right\}$ within at most $T_c$ time units. Let $\bm{\Phi}(t,\bm{x}_0)$ be a solution of~\eqref{eq:sys}, with $\bm{x}_0\in\left\{\bm{x}\in\mathbb{R}^n:V(\bm{x})\geq \alpha_2(\mu)\right\}$, i.e. $V(\bm{x}_0)\geq\alpha_2(\mu)$. From~\eqref{eq:Vineq_ub} and following similar arguments as in Theorem~\ref{thm:predefined}, $\kappa(V(\bm{\Phi}(t,\bm{x}_0)))\leq \left[\kappa(V(\bm{x}_0))^{1-p}-\frac{t}{T_c}\right]^{\frac{1}{1-p}}$, for $t\in\left[0,T_c(\kappa(V(\bm{x}_0))^{1-p}-\kappa(\alpha_2(\mu))^{1-p})\right]$. Hence, $\kappa(V(\bm{\Phi}(t,\bm{x}_0)))\leq\kappa(\alpha_2(\mu)) \Longleftrightarrow V(\bm{\Phi}(t,\bm{x}_0))\leq\alpha_2(\mu)$ for all $t\geq T_c(\kappa(V(\bm{x}_0))^{1-p}-\kappa_2(\mu)^{1-p})$, and consequently for all $t\geq T_c$.

Furthermore, note that $V(\bm{x})\leq\alpha_2(\mu) \Rightarrow \alpha_1(\norm{\bm{x}})\leq\alpha_2(\mu) \Longleftrightarrow  \norm{\bm{x}}<\alpha_1^{-1}(\alpha_2(\mu))$. Hence, $\norm{\bm{\Phi}(t,\bm{x}_0)}\leq\alpha_1^{-1}(\alpha_2(\mu))$ for all $t\geq T_c$. 

Moreover, if $V(\bm{x})=\alpha(\norm{\bm{x}})$, one can select $\alpha_1=\alpha_2=\alpha$, and the result is obtained.
\end{proof}

\section{Application: Lyapunov-based predefined-time controller design} \label{sec:control}

\subsection{Problem statement}

Consider the following affine control system: 
\begin{equation}\label{eq:cont_sys}
\dot{\bm{x}} = \bm{f}(\bm{x}) + \bm{B}(\bm{x})\bm{v} + \bm{\delta}(\bm{x},t)
\end{equation}
where $\bm{x}:\mathbb{R}_{\geq 0}\to\mathbb{R}^n$ is the system state, $\bm{v}\in\mathbb{R}^m$ is the control input,
$\bm{\delta}:\mathbb{R}^n\times\mathbb{R}_{\geq 0}\rightarrow\mathbb{R}^n$ is a disturbance vector that includes plant parameter variations and external unknown perturbations, and $\bm{B}:\mathbb{R}^n\to\mathbb{R}^{n\times n}$ is continuous and such that $\text{rank }\bm{B}(\bm{x})=m$ for all $\bm{x}\in\mathbb{R}^n$.

The objective is to design a feedback control input $\bm{v}$ such that the trajectories of~\eqref{eq:cont_sys} reach (a vicinity of) the manifold
\begin{equation}\label{eq:constraint}
\bm{s}(\bm{x},t)=0,
\end{equation}
where $\bm{s}:\mathbb{R}^n\times\mathbb{R}_{\geq 0}\rightarrow\mathbb{R}^m$ is a smooth mapping, in an arbitrarily selected time $T_c\in\mathbb{R}_+$ and remain there for all $t\geq T_c$.

\begin{remark} There are several important control problems which take the form of system~\eqref{eq:cont_sys} subject to~\eqref{eq:constraint}. For instance,
\begin{itemize}
\item[\textit{(i)}] An output tracking problem, where $\bm{s}(\bm{x},t)$ represents the output tracking error vector; the relative degree of each output component with respect to the control input is equal to one, and the system~\eqref{eq:cont_sys} with respect to \eqref{eq:constraint} is minimum phase~\cite{Ley-Rosas2016}. The control objective is to ensure the tracking error be predefined time stable.
\item[\textit{(ii)}] A sliding mode (SM) control design problem, where $\bm{s}(\bm{x},t)=0$~\eqref{eq:constraint} represents a sliding manofold with a desired SM motion~\cite{Utkin1992}. In this case, the objective is to induce the SM on the designed manifold in predefined-time (predefined-time reaching phase).
\item[\textit{(iii)}] An optimization problem solved by dynamic networks, where $\bm{s}(\bm{x},t)$ is a variable which expresses the error in the satisfaction of equality constraints~\cite{Utkin1992}.
\end{itemize}
\end{remark}

The time derivative of $\bm{s}(\bm{x},t)$ is
\begin{equation}\label{ueq1}
\dot{\bm{s}} = \bm{G}(\bm{x},t)\bm{f}(\bm{x})+ \bm{G}(\bm{x},t)\bm{B}(\bm{x})\bm{v} +
\bm{G}(\bm{x},t)\bm{\delta}(\bm{x},t) +\frac{\partial \bm{s}(\bm{x},t)}{\partial t}.
\end{equation}
Then, assuming that $\text{rank} [\bm{G}(\bm{x},t)\bm{B}(\bm{x})]=m$, for all $\bm{x}\in\mathbb{R}^n$ and $t\geq 0$, the control $\bm{v}$ is chosen as
\begin{equation}\label{ueq2}
\bm{v} = -[\bm{G}(\bm{x},t)\bm{B}(\bm{x})]^{-1}\bigg[\bm{G}(\bm{x},t)\bm{f}(\bm{x}) +\frac{\partial \bm{s}(\bm{x},t)}{\partial t}+\bm{u}\bigg],
\end{equation}
where $\bm{u}\in\mathbb{R}^m$ is a virtual control input and $\bm{G}(\bm{x},t)=\frac{\partial \bm{s}(\bm{x},t)}{\partial \bm{x}}$.

Substituting~\eqref{ueq2} in~\eqref{ueq1} results in
\begin{equation}\label{eq:reduced_sys}
\dot{\bm{s}}=\bm{u}+\bm{\Delta}(\bm{x},t), \ \ \bm{s}(\bm{x}_0,0)=\bm{s}_0,
\end{equation}
where $\bm{\Delta}(\bm{x},t)=\bm{G}(\bm{x},t)\bm{\delta}(\bm{x},t)$, which is assumed to be globally bounded by $\sup_{(\bm{x},t)\mathbb{R}^n\in\times\mathbb{R}_{\geq0}}\norm{\bm{\Delta}(\bm{x},t)}\leq\delta$ with $0\leq\delta<\infty$ a known constant.

From a control design point of view, the perturbation term $\bm{\Delta}(\bm{x},t)$ can only be completely rejected by a discontinuous control term (like the unit-vector controller $\frac{\bm{s}}{\norm{\bm{s}}}$), given that it is only restricted to be bounded (no conditions of smoothness, Lipschitz continuity neither continuity are assumed). However, such a discontinuous control term might deteriorate the components of a real physical system due to high-frequency oscillations, or might even be impossible to implement due to limited actuator response.

A solution would be to sacrifice the exact convergence to the manifold $\bm{s}(\bm{x},t)=0$~\eqref{eq:constraint} in order to obtain a continuous controller (like the continuous approximation of the unit-vector controller $\frac{\bm{s}}{\norm{\bm{s}}+\epsilon}$, with $\epsilon>0$). In this case, as pointed out in Section~\ref{sec:prelim}, it can be ensured that the trajectories converge to a vicinity of $\bm{s}(\bm{x},t)=0$ (see Eq. \eqref{eq:constraint}).

Based on the above, the objective is to design the virtual control input $\bm u$ as a feedback control law $\bm u=\bm u(\bm s)$ that:
\begin{itemize}
\item enforces predefined-time stability to the origin of~\eqref{eq:reduced_sys}, obtaining a discontinuous controller; or
\item enforces the solutions of~\eqref{eq:reduced_sys} to be predefined-time ultimately bounded with predefined bound, obtaining a continuous controller.
\end{itemize}

\subsection{Proposed solution}
The proposed solution is a corollary of Theorems~\ref{thm:predefined} and \ref{thm:predefined_ub}.

\begin{corollary} \label{cor:cont_cont} Consider system~\eqref{eq:reduced_sys}. Selecting $\bm{u}$ as
\begin{equation}\label{eq:cont_cont}
\bm{u}=-\frac{1}{(1-\rho_2)\rho_1}\frac{\kappa(\norm{\bm{s}})^{\rho_2}}{\kappa'(\norm{\bm{s}})}\barpowvec{\bm{s}}^{0}-\rho_3\frac{\bm{s}}{\norm{\bm{s}}+\rho_4},
\end{equation}
where $\kappa\in\mathcal{K}^1$ is such that $\kappa':\mathbb{R}_{\geq 0}\to\Bar{\mathbb{R}}_{+}$, $\rho_1>0$, $0\leq\rho_2<1$, $\rho_3>\delta$ and $\rho_4\geq0$, the trajectories of the closed-loop system~\eqref{eq:reduced_sys}-\eqref{eq:cont_cont} are predefined-time ultimately bounded with predefined bound. In fact, for any $T_c,b\in\mathbb{R}_+$ and $\bm{s}_0\in\mathbb{R}^m$, the solution $\bm{\Phi}(t,\bm{s}_0)$ of~\eqref{eq:reduced_sys}-\eqref{eq:cont_cont} satisfies $\norm{\bm{\Phi}(t,\bm{s}_0)}\leq b=\frac{\delta\rho_4}{\rho_3-\delta}$ for all $t \geq T_c=\rho_1$.

Moreover, if  $\rho_4=0$, then the origin $\bm{s}=\bm{0}$ of the closed-loop system~\eqref{eq:reduced_sys}-\eqref{eq:cont_cont} is predefined-time stable.
\end{corollary}
\begin{proof} Consider the the continuous, positive definite and radially unbounded function $V(\bm{s})=\norm{\bm{s}}$, and let $T_c,b\in\mathbb{R}_+$. The time-derivative of $V$ along the trajectories of the closed-loop system~\eqref{eq:reduced_sys}-\eqref{eq:cont_cont} is $\dot{V}(\bm{s})=\barpowvec{\bm{s}^T}^{0}\left[-\frac{1}{(1-\rho_2)\rho_1}\frac{\kappa(\norm{\bm{s}})^{\rho_2}}{\kappa'(\norm{\bm{s}})}\barpowvec{\bm{s}}^{0}\\-\rho_3\frac{\bm{s}}{\norm{\bm{s}}+\rho_4}+\bm{\Delta}(\bm{x},t)\right]$. Note that:
\begin{itemize}
\item[\textit{(a)}] the continuous approximation of the unit-control term satisfies $\frac{\bm{s}}{\norm{\bm{s}}+\rho_4}=\barpowvec{\bm{s}}^0-\frac{\rho_4\barpowvec{\bm{s}}^0}{\norm{\bm{s}}+\rho_4}$;
\item[\textit{(b)}] the product $\barpowvec{\bm{s}^T}^0\barpowvec{\bm{s}}^0=1$;
\item[\textit{(c)}] by the Cauchy-Schwarz inequality $\barpowvec{\bm{s}^T}^{0}\bm{\Delta}\leq\delta$.
\end{itemize}
Therefore, considering \textit{(a), (b)} and \textit{(c)}, the time derivative $\dot{V}(\bm{s})$ results in
\begin{equation*}
\dot{V}(\bm{s})\leq -\frac{1}{(1-\rho_2)\rho_1}\frac{\kappa(V(\bm{s}))^{\rho_2}}{\kappa'(V(\bm{s}))}-\left(\rho_3-\delta-\frac{\rho_3\rho_4}{\norm{\bm{s}}+\rho_4}\right).
\end{equation*}
Since $\rho_3-\delta-\frac{\rho_3\rho_4}{\norm{\bm{s}}+\rho_4}\geq 0 \iff \norm{\bm{s}}\geq\frac{\delta\rho_4}{\rho_3-\delta}$, then
\begin{equation}\label{eq:proof}
\dot{V}(\bm{s})\leq -\frac{1}{(1-p)T_c}\frac{\kappa(V(\bm{s}))^{p}}{\kappa'(V(\bm{s}))}, \quad \text{for } \norm{\bm{s}}\geq\mu,
\end{equation}
with $T_c=\rho_1$, $\mu=\frac{\delta\rho_4}{\rho_3-\delta}$ and $p=\rho_2$.

From the above and using Theorem~\ref{thm:predefined_ub}, the solutions of the closed-loop system~\eqref{eq:reduced_sys}-\eqref{eq:cont_cont} satisfy $\norm{\bm{\Phi}(t,\bm{s}_0)}\leq b=\mu=\frac{\delta\rho_4}{\rho_3-\delta}$ for all $t \geq T_c=\rho_1$.

Moreover, if $\rho_4=0\Longleftrightarrow\mu=0$, inequality~\eqref{eq:proof} holds for all $\bm{s}\in\mathbb{R}^m$. Hence, by Theorem~\ref{thm:predefined}, the origin $\bm{s}=\bm{0}$ of the closed-loop system~\eqref{eq:reduced_sys}-\eqref{eq:cont_cont} is predefined-time stable.
\end{proof}

It is worth to notice that the predefined time $T_c=\rho_1$ and the predefined bound $b=\frac{\delta\rho_4}{\rho_3-\delta}$ can be selected independently since they depend on different parameters.


\begin{example} Let $\bm{x}=[x_1 \quad x_2]^T\in\mathbb{R}^2$ be a point in the plane whose dynamics is given by $\dot{\bm{x}}=\bm{u}$, where $\bm{u}=[u_1 \quad u_2]^T\in\mathbb{R}^2$ is the feedback control signal to be designed so that the point $\bm{x}$ tracks a desired reference trajectory $\bm{r}(t)=[r_1(t) \quad r_2(t)]^T:\mathbb{R}_{\geq0}\to\mathbb{R}^2$. Whereas the reference trajectory signal is assumed to be known, its derivative $\dot{\bm{r}}(t)$ is assumed to be unknown but bounded of the form $\sup_{t\in\mathbb{R}_{\geq0}}\norm{\dot{{\bm{r}}}(t)}\leq\delta$.
This model is a classic example in teleoperation tasks, such as haptic interfaces, remote command of manipulators, land, aerial and underwater robots, to name a few, where the position reference is given in real-time by the user, but the desired velocity reference is unknown.
The dynamics of the error variable, $\bm{s}(\bm{x},t)=\bm{x}-\bm{r}(t)$, is then $\dot{\bm{s}}=\bm{u}-\dot{\bm{r}}(t)$. Under all the above assumptions, the feedback control signal $\bm{u}$ can be designed as~\eqref{eq:cont_cont} in Corollary~\ref{cor:cont_cont}.

For simulation purposes, the function $\kappa(r)$ in~\eqref{eq:cont_cont} is selected as $\kappa(r)=I\left(m_{\rho_7},m_{\rho_8},\frac{\rho_6 r^{\rho_8-\rho_7}}{\rho_6 r^{\rho_8-\rho_7}+\rho_5}\right),$ with $m_{\rho_7}=\frac{1-\rho_9\rho_7}{\rho_8-\rho_7}>0$, $m_{\rho_8}=\frac{\rho_9\rho_8-1}{\rho_8-\rho_7}>0$, $\rho_5,\rho_6,\rho_9>0$ and $0<\rho_9\rho_7<1<\rho_9\rho_8$. Furthermore, the reference signal is selected as $\bm{r}(t)=[\cos(2\pi t) \quad \sin(2\pi t)]^T$ (i.e., the point $\bm{x}$ is required to follow a circumference of radius $1$), whose derivative norm is $\norm{\dot{\bm{r}}(t)}=2\pi=\delta$. Moreover, the following simulations are conducted using the Euler integration method, with a fundamental step size of $1\times10^{-5}$~s. Finally, the initial conditions of the system are set as $\bm{x}_0=[x_{10} \quad x_{20}]=x_0[1 \quad 1]$, with $x_0$ taking the values of $10^1$, $10^3$ and $10^{21}$.

\textit{Part I}: assume that the tracking error $\bm{s}$ is required to reach a vicinity of the origin with a radius of measure $b=0.01$ in at most $T_c=1$ time units. To this end, the parameters of controller~\eqref{eq:cont_cont} are selected as $\rho_1=1$, $\rho_2=0$, $\rho_3=4\pi$, $\rho_4=0.01$, $\rho_5=1$, $\rho_6=1$, $\rho_7=0.9$, $\rho_8=1.1$, and $\rho_9=1$. Note that, with this parameter selection, controller~\eqref{eq:cont_cont} is continuous.

\begin{figure}[!t]
\centering
\includegraphics[width=3in]{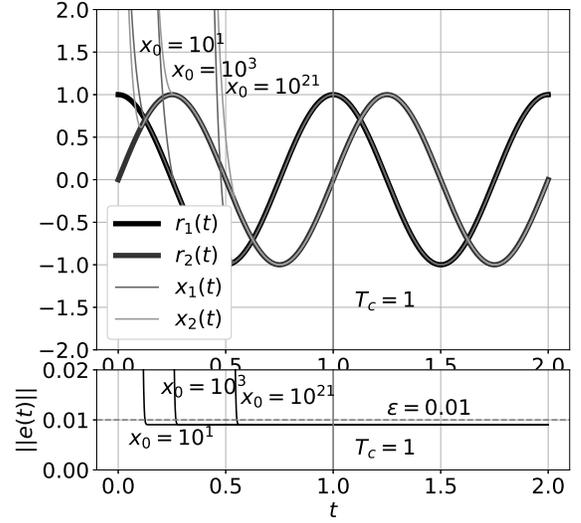}
\caption{Several trajectories of the state variables $x_1,x_2$ vs. $t$, and $\norm{\bm{e}}$ vs. $t$. Continuous controller.}
\label{fig:control_cont}
\end{figure}

\begin{figure}[!t]
\centering
\includegraphics[width=3in]{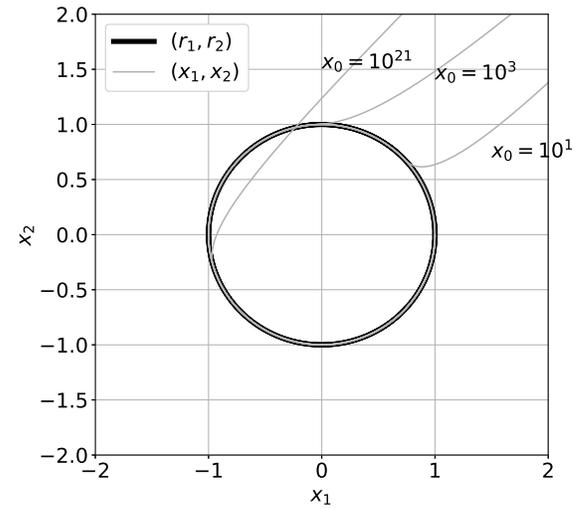}
\caption{Several trajectories of the point $(x_1,x_2)$ in the plane. Continuous controller.}
\label{fig:plane_cont}
\end{figure}

\textit{Part II}: assume that the tracking error $\bm{s}$ is required to reach the origin in at most $T_c=1$ time units. To this end, the parameters of  controller~\eqref{eq:cont_cont} are selected as $\rho_1=1$, $\rho_2=0$, $\rho_3=4\pi$, $\rho_4=0$, $\rho_5=1$, $\rho_6=1$, $\rho_7=0.9$, $\rho_8=1.1$, and $\rho_9=1$. Note that, with this parameter selection, controller~\eqref{eq:cont_cont} is discontinuous.

\begin{figure}[!t]
\centering
\includegraphics[width=3in]{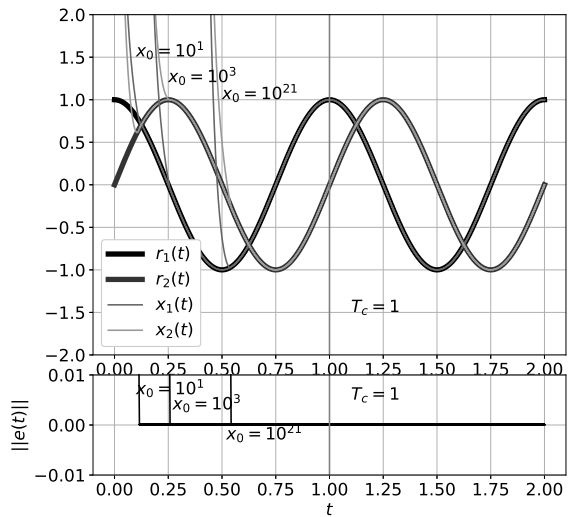}
\caption{Several trajectories of the state variables $x_1,x_2$ vs. $t$, and $\norm{\bm{e}}$ vs. $t$. Discontinuous controller.}
\label{fig:control_disc}
\end{figure}

\begin{figure}[!t]
\centering
\includegraphics[width=3in]{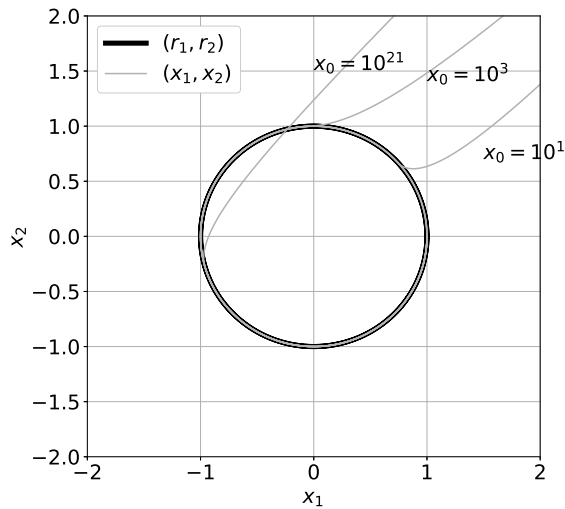}
\caption{Several trajectories of the point $(x_1,x_2)$ in the plane. Discontinuous controller.}
\label{fig:plane_disc}
\end{figure}

Figs.~\ref{fig:control_cont} and \ref{fig:control_disc} show the state variables and the reference signal over time, respectively, for the continuous ($\rho_4=0.01$) and the discontinuous ($\rho_4=0$) cases. They also show the tracking error norm $\norm{\bm{s}}$ for a better appreciation of the convergence to the $0.01$-vicinity of the origin, in the continuous case, or the convergence to the origin, in the discontinuous case. Figs.~\ref{fig:plane_cont} and \ref{fig:plane_disc} show the trajectories of the point $(x_1,x_2)$ in the plane, respectively, for the continuous and the discontinuous cases.
\end{example}

\section{Conclusion}
This technical note was devoted to the study of sufficient Lyapunov-like conditions to ensure a class of dynamic systems to exhibit a predefined time stability property. The introduction of class $\mathcal{K}^1$ allowed to establish equivalence with previous Lyapunov-like theorems for predefined-time stability for autonomous systems ~\cite{Sanchez-Torres2018,Sanchez-Torres2018a,Aldana-Lopez2019,Jimenez-Rodriguez2018a}. Moreover, the derived Lyapunov theorem was extended for analyzing predefined-time ultimate boundedness with predefined bound.

On the other hand, the developed framework was used to design a class of robust controllers for uncertain affine control systems. This class of controllers can be continuous, providing predefined-time ultimate boundedness of the solutions, or discontinuous, providing predefined-time stability to some desired manifold. Finally, the theoretical findings were validated through a numerical simulation, which reveals the effectiveness of the proposed control scheme.

As future work, the Lyapunov-like conditions introduced in Theorems~\ref{thm:predefined}-\ref{thm:predefined_ub}, require further research to exploit them for the controller design of particular classes of nonlinear systems.



\section*{Acknowledgment}

Esteban Jim\'enez acknowledges to CONACYT--M\'exico for the D.Sc. scholarship number 481467 and the project 252405.
Aldo Jonathan acknowledges to CONACYT--M\'exico for the Project C\'atedras 1086 ``Ambientes Inteligentes''.

\bibliographystyle{IEEEtran}
\bibliography{predefined}

\end{document}